\documentclass[12pt,reqno]{amsart}
\usepackage{hyperref}
\hypersetup{
  colorlinks   = true, 
  urlcolor     = blue, 
  linkcolor    = blue, 
  citecolor   = red 
  }
\usepackage{textcomp}
\usepackage[mathscr]{eucal}
\usepackage{amsmath,amsfonts,amsthm,amscd,amssymb,scalefnt}
\usepackage{color}
\usepackage{graphics}
\usepackage{float}
\usepackage{graphicx}
\usepackage{epstopdf}
\usepackage{subfigure}
\usepackage{subfloat}
\newtheorem{theorem}{Theorem}
\theoremstyle{plain}

\newtheorem{corollary}{Corollary}

\newtheorem{definition}{Definition}

\newtheorem{remark}{Remark}

\numberwithin{equation}{section}
\begin{document}
\title[Girard Type Theorems for de Sitter Triangles with non-null Edges]{Girard Type Theorems for de Sitter Triangles with non-null Edges}
\author{Baki Karliga}
\address{Gazi University \\
Science Faculty\\
Department of Mathematics\\
06370 Teknikokullar Ankara/TURKEY}
\email{\ \ karliaga@gazi.edu.tr}
\author{Umit Tokeser}
\address{Kastamonu University \\
Science and Arts Faculty \\
Department of Mathematics\\
Kastamonu/TURKEY}
\email{ \ utokeser@kastamonu.edu.tr}
\subjclass[2010]{Primary 51B20, 51M25, 97G30; Secondary 53A35, 83C80, 51P05}
\keywords{Girard's Theorem, triangle, de Sitter triangle}
\begin{abstract}
Girard's Theorem subjects to the area depending interior angles of a spherical triangle. In this paper, we introduce to its analogues for proper de Sitter triangles with non-null edges.
\end{abstract}
\maketitle
\section{Introduction}

If something exerts a force on a particle, then this phenomenon is called
gravity. A free moving or falling particle follows a geodesics in a
space-time. Thus, the geometry of space-time is modified by the sources of
gravity. At far away, the sources of gravity is called dark energy reveals de Sitter rather than Minkowski space-time. Thus, de Sitter space is a suitable model for the universe as it consistently explains its structure.

A space (space-time) is called {\it flat or curved} if it has zero or not zero curvature.A geodesic in a flat space( space-time)is always straight lines while in a curved space( space-time) is always curved line. Although a geodesic of curved space has only space-like \emph{causal direction}, a geodesic of curved space-time has one of three different \emph{causal directions} which are space-like, time-like and light-like. While it is only hyperbola (great circle) in hyperbolic (spherical) curved space,a geodesic of de Sitter curved space-time is one of  three curves which are ellipse, hyperbola and straight line. Thus, geodesics in de Sitter space-time is quite different and rich from the spherical and hyperbolic curved space. In view of triangular shapes, de Sitter space is richer and more universal than Euclidean, Minkowskian, spherical and hyperbolic spaces.

For a de Sitter triangle, the plane spanned by two tangent vectors at a vertex is called  {\bf angle plane} of that vertex. If the angle plane of a vertex is space-like, light-like or time-like, then angle is called space-like, light-like or time-like.
\newline
The plane containing an edge of a de Sitter triangle is called {\bf edge plane} of that edge. If an edge plane is space-like, light-like or time-like, then edge is called space-like, light-like, or time-like.

In space-time geometry, triangles can be classified according to causal type of its angles and edges \cite{2}. In de Sitter space, the triangle classification according to causal type of edges is given by Asmus \cite{3}. He showed that there are ten different triangles and only four of them (spatiolateral, tempolateral, chorosceles and chronosceles) have a polar triangle.

 The complex valued pseudoangle, and the complex valued area depending on interior pseudo-angles of a de Sitter triangle is given in \cite{4}. Peiro, in \cite[Theorem 1.1]{5}, gave the relationship the dihedral angle with the angle between normals of two edge planes.

 The area depending on interior angles $ \alpha, \beta, \gamma$ of a triangle on unit sphere is given by Girard's Theorem as$ (\alpha+\beta+\gamma)-\pi$ \cite{1}. Hyperbolic analogue of Girard's Theorem, known Lambert's Theorem, gives the area depending on interior angles $ \alpha, \beta, \gamma $ of a triangle on unit hyperbolic plane as $\pi- (\alpha+\beta+\gamma)$ \cite{6}. The complex valued analogue of Girard's theorem for de Sitter triangles with non-null edge is introduced by Dzan\cite{4}.

By introducing the relationship with complex valued pseudo-angle and angle in ${\mathbb{R}}_{1}^{3}$ , we give Girard's theorems subjects to area depending interior angles of a contractible spatiolateral, tempolateral, chrosceles and chronosceles de Sitter triangles.

\section{Geodesics and Triangles in de Sitter Space $S_{1}^{2}$}
If $w=q-\langle p,q\rangle p$~and~$W=sp\left\{ p,q\right\}$ for $p,q\in S_{1}^{2}$, then
by \cite{7} and \cite{8}, one can easily prove the following results.
\begin{theorem}\label{teo1}
\hspace*{\fill}
\begin{enumerate}
\item $\left\vert \left\langle p,q\right\rangle \right\vert
<1\Leftrightarrow w  $ is space-like\smallskip
\item $\left\vert\left\langle p,q\right\rangle\right\vert
>1\Leftrightarrow w $ is time-like\smallskip
\item $\left\vert \left\langle p,q\right\rangle \right\vert
=1\Leftrightarrow w$ is null.\smallskip
\end{enumerate}
\end{theorem}
\begin{theorem}\label{teo2}
\hspace{1cm}
\begin{enumerate}
\item $\left\vert\left\langle p,q\right\rangle\right\vert
<1\Leftrightarrow W$ space-like\smallskip
\item$\left\vert\left\langle p,q\right\rangle\right\vert
>1\Leftrightarrow W$ time-like\smallskip
\item $\left\vert\left\langle p,q\right\rangle\right\vert
=1\Leftrightarrow W$ null.\smallskip
\end{enumerate}
\end{theorem}
\begin{theorem}\label{teo3}
Let $p, q\in S_{1}^{2}$ ~$\overrightarrow{pq}$ is light-like if and only if
$\left\langle p,q\right\rangle =1$.
\end{theorem} 
\newpage
\begin{theorem}\label{teo4}
\hspace*{\fill}
\begin{enumerate}
\item $W$ is time-like and $\left\langle p,q\right\rangle >1$ if and only if $p$ and $q$ is on the same part of hyperbola
$W\cap S_{1}^{n} $. \medskip
\item $W$ is time-like and $\left\langle
p,q\right\rangle <-1$ if and only if $p$ and $q$ is on the different part of
the hyperbola $W\cap S_{1}^{n}$. \medskip
\item $W$ is space-like if and only if $\left\vert\left\langle p,q\right\rangle\right\vert < 1$.\medskip
\item $W$ is null if and only if $\left\vert
\left\langle p,q\right\rangle \right\vert =1$. \medskip
\end{enumerate}
\end{theorem}
\begin{theorem}\label{teo5}
Let $l$ be geodesic segment bounded by $p$ and $q$, then \smallskip
\begin{enumerate}
\item $l$ is hyperbola part if and only if $\left\langle p,q\right\rangle >1$. 
\smallskip
\item $l$ is ellipse part if and only if $\left\vert \left\langle
p,q\right\rangle \right\vert <1.$
\item $l$ is null line segment if and only if $\left\langle p,q\right\rangle
=1.$
\item $l$ is impossible line segment if and only if $\left\langle
p,q\right\rangle <-1$.
\end{enumerate}
\end{theorem}
A generalized de Sitter triangle $\Omega $ can be seen as follows in Asmus \cite{3}:
\begin{enumerate}
\item If $\Omega \subset H^{2}, \Omega $ is called \textbf{hyperbolic triangle}
\item If $\Omega \subset (-H^{2}), \Omega $ is called \textbf{antipodal
hyperbolic triangle}
\item If $\Omega \subset S_{1}^{2}, \Omega $is called \textbf{proper de
Sitter triangle}
\item Otherwise, $\Omega $ is called \textbf{strange triangle}
\item If at least one edge is empty, then, $\Omega $ is called\textbf{\
impossible triangle}.
\end{enumerate}
\smallskip

Let $i,j$ and $k$ be the number of spacelike, timelike and lightlike edges of a de Sitter triangle $\overset{k}{_{i}\triangle _{j}}$.

$\overset{3}{_{0}\triangle _{0}}$~,~$\overset{1}{_{1}\triangle _{1}}$~,~$\overset{2}{_{1}\triangle _{0}}$~,~$\overset{2}{_{0}\triangle _{1}}$~,~$\overset{1}{_{0}\triangle _{2}}$~,~$\overset{1}{_{2}\triangle _{0}}$ are the proper de Sitter triangles with null edges, and are called Lucilateral, Multiple, Photosceles with space-like base, Photosceles with time-like, Bimetrical Chronosceles, Bimetrical Chorosceles Triangle, respectively \mbox{(see Figure \ref{fg4301})}.
\begin{figure}[ht]
\centering
\subfigure[]{\includegraphics[width=3cm,height=3cm,keepaspectratio]{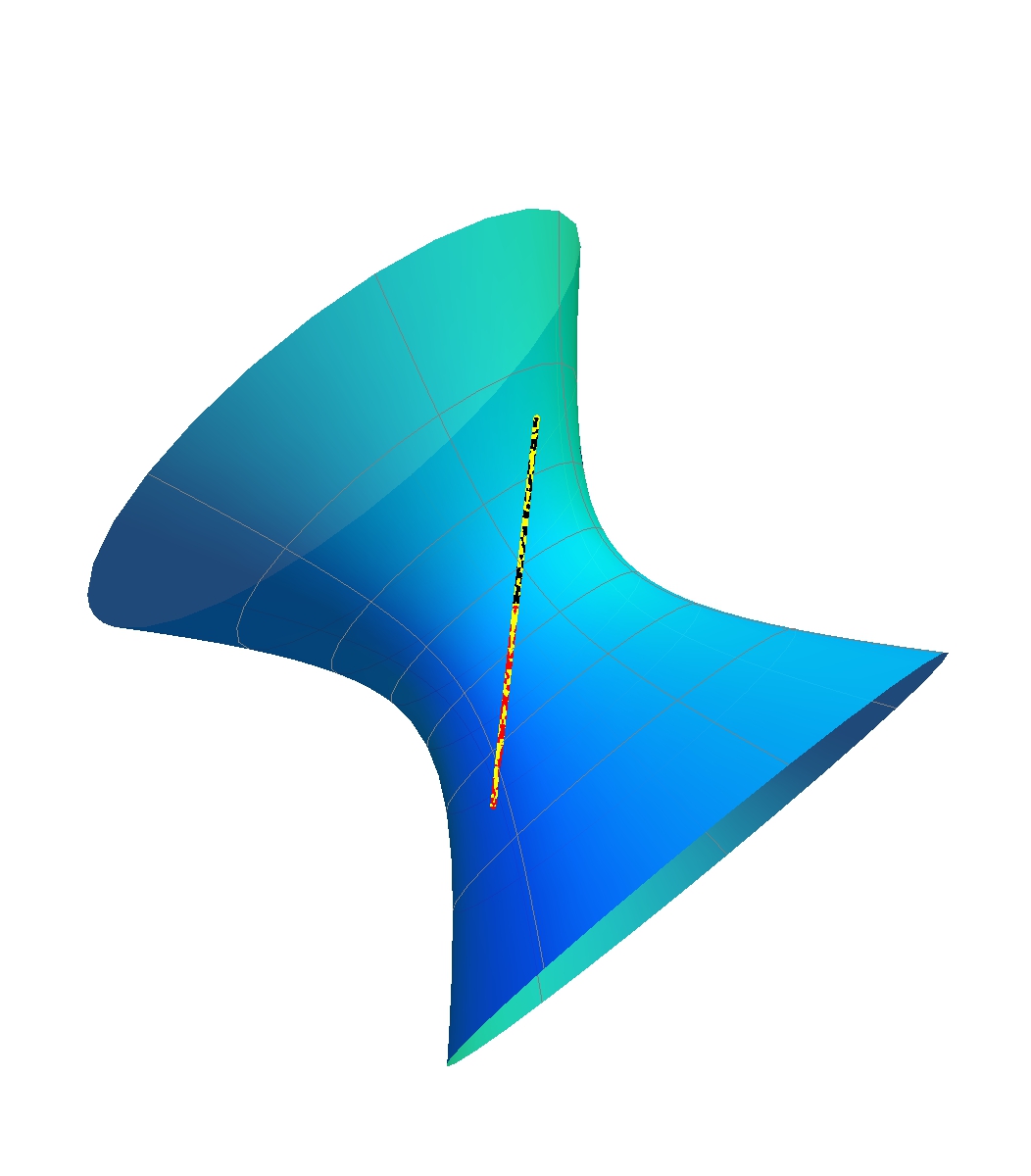}
\label{fg4301a}}
\quad
\subfigure[]{\includegraphics[width=4cm,height=3cm,keepaspectratio]{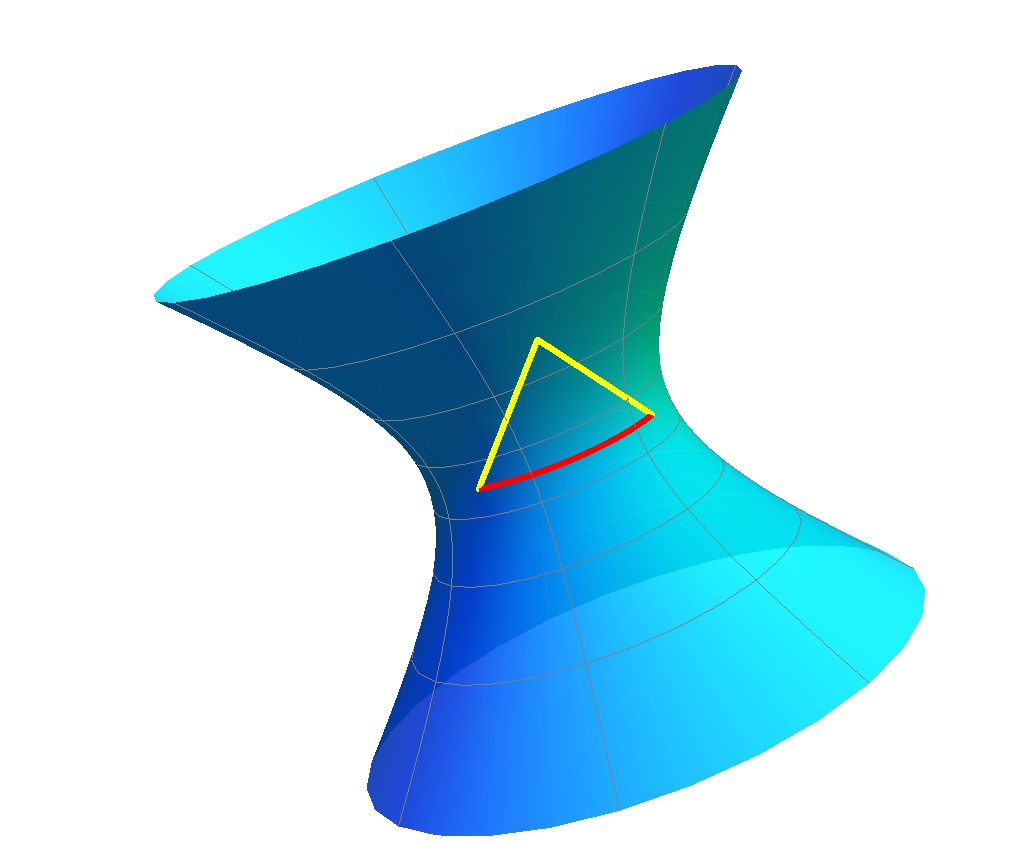}
\label{fg4301b}}
\quad
\subfigure[]{\includegraphics[width=4cm,height=3cm,keepaspectratio]{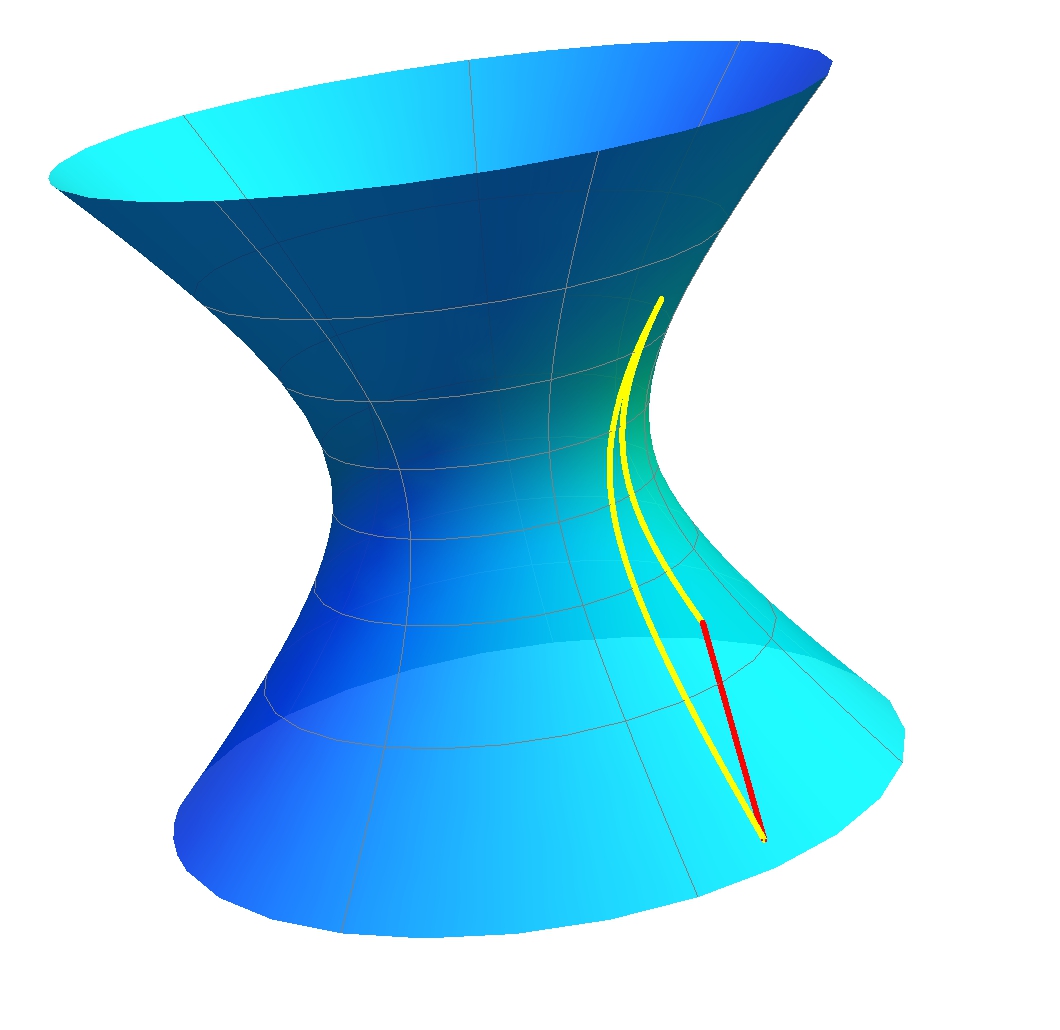}
\label{fg4301c}}
\quad
\subfigure[]{\includegraphics[width=4cm,height=3cm,keepaspectratio]{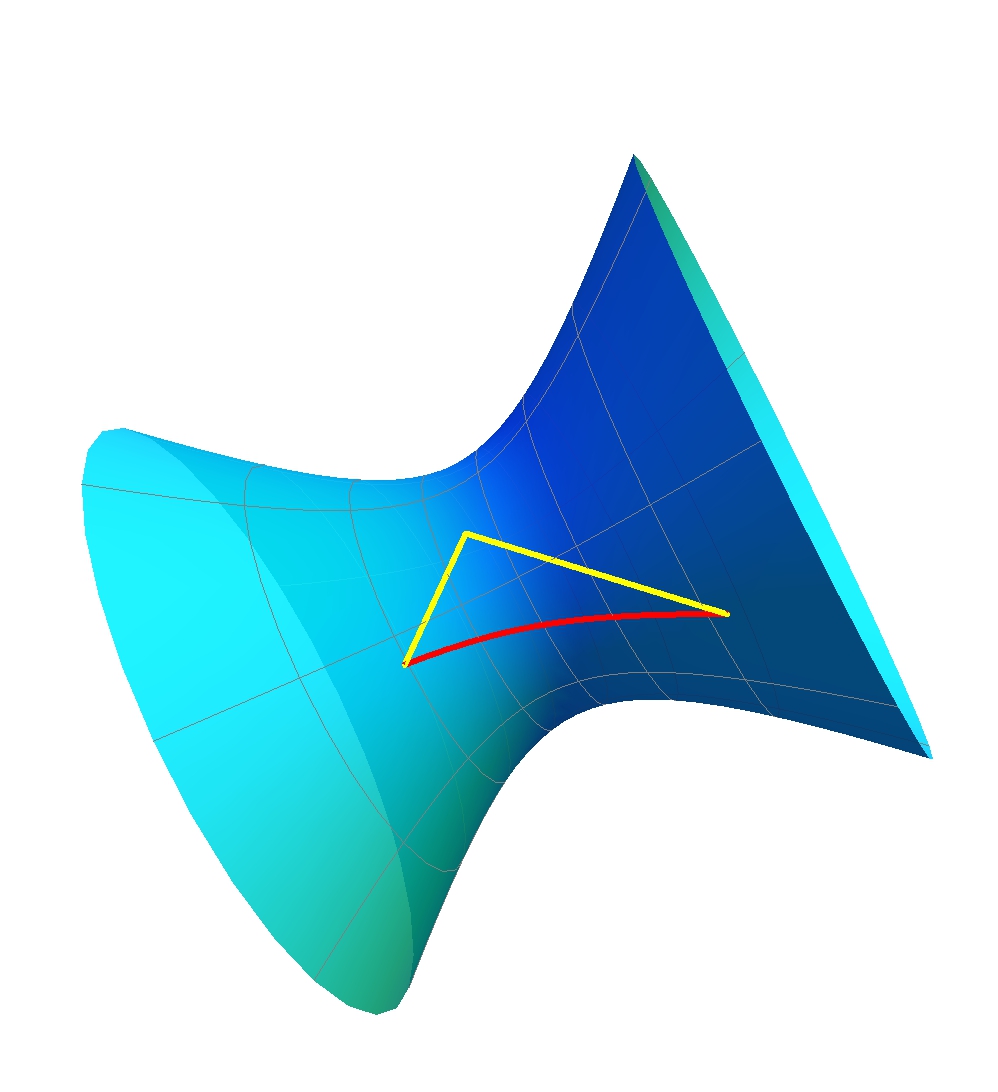}
\label{fg4301d}}
\quad
\subfigure[]{\includegraphics[width=4cm,height=3cm,keepaspectratio]{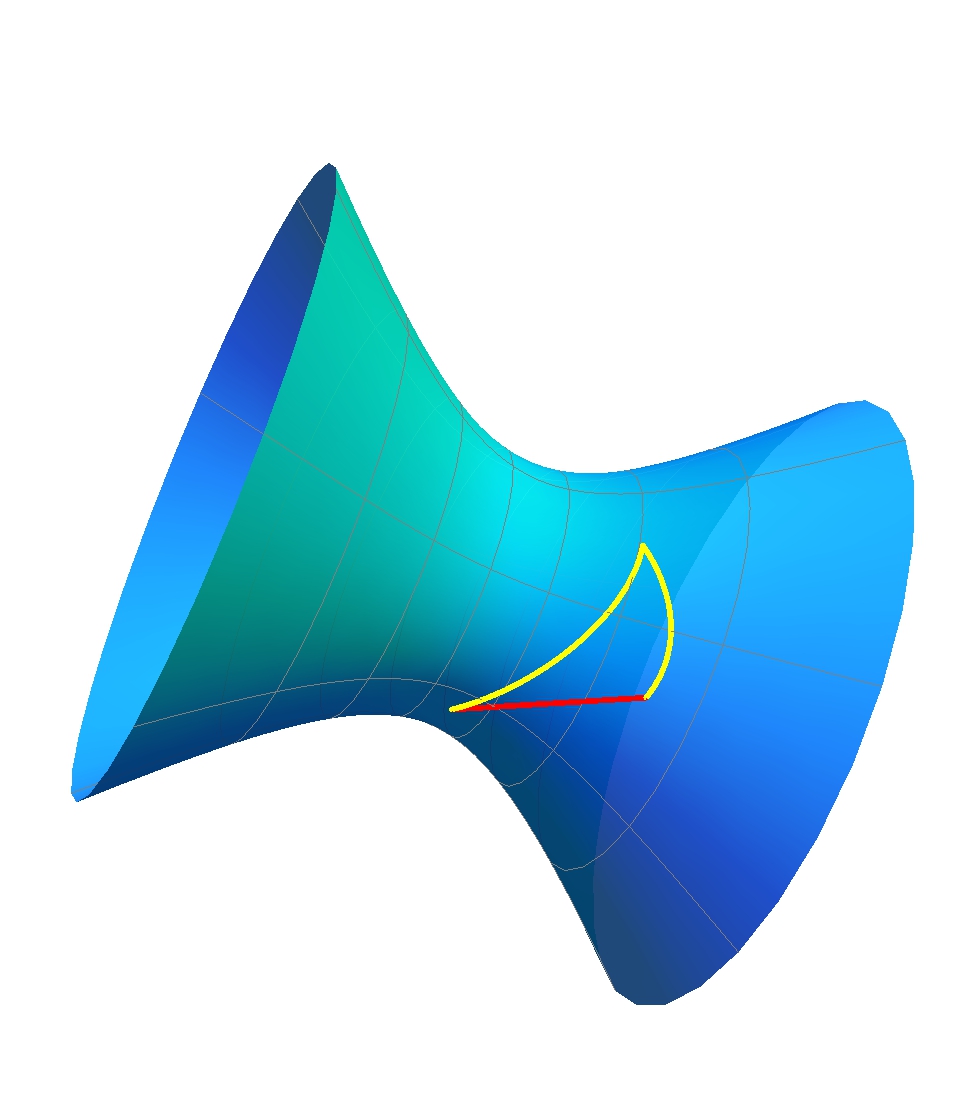}
\label{fg4301e}}
\quad
\subfigure[]{\includegraphics[width=4cm,height=3cm,keepaspectratio]{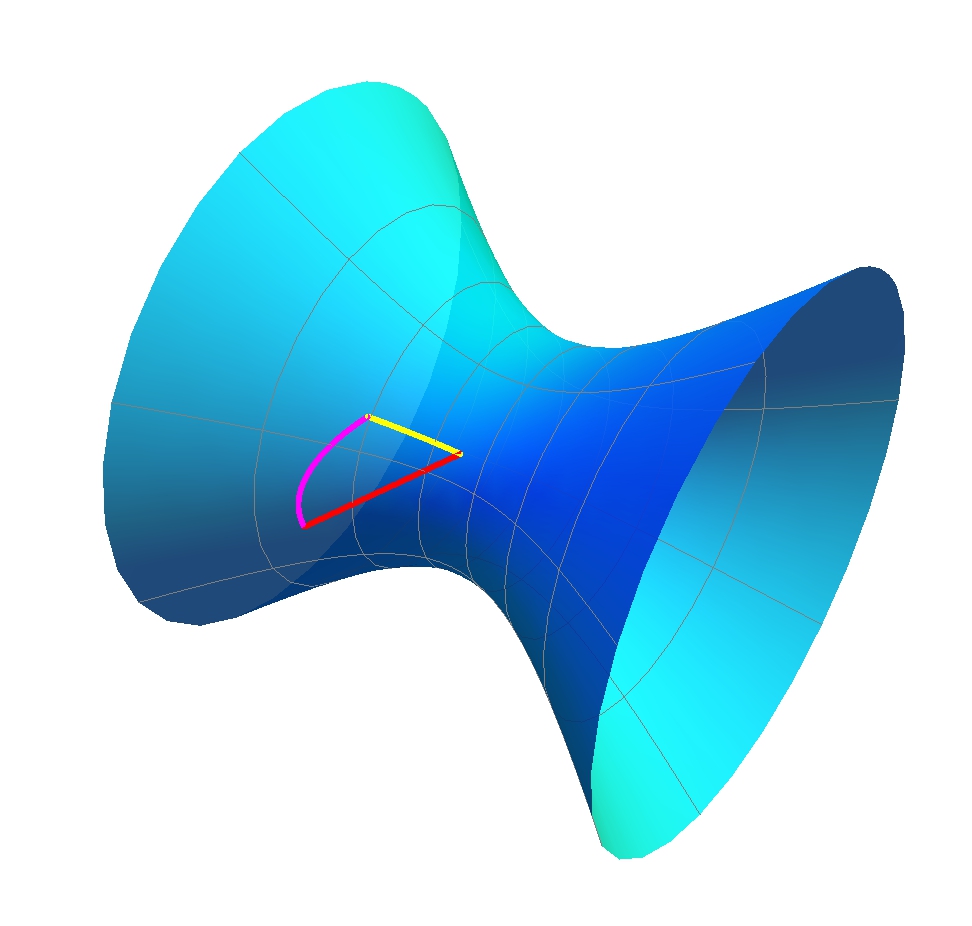}
\label{fg4301f}}
\caption{\scalefont{0.8}{(a) lucilateral (b) photosceles with space-like base (c) bimetrical chronosceles (d) photosceles with time-like base (e)bimetrical chorosceles (f) multiple triangle}}
\label{fg4301}
\end{figure}

$\overset{0}{_{3}\triangle _{0}}$~,~$\overset{0}{_{0}\triangle _{3}}$~,~$\overset{0}{_{2}\triangle _{1}}$~and~$\overset{0}{_{1}\triangle_{2}}$ are the proper de Sitter triangles with non-null edges, and are called Spatiolateral, Tempolateral, Chorosceles,  Chronosceles Triangle, respectively (see Figure \ref{fg4302}).
\begin{figure}[ht]
\centering
\subfigure[]{\includegraphics[width=3cm,height=3cm,keepaspectratio]{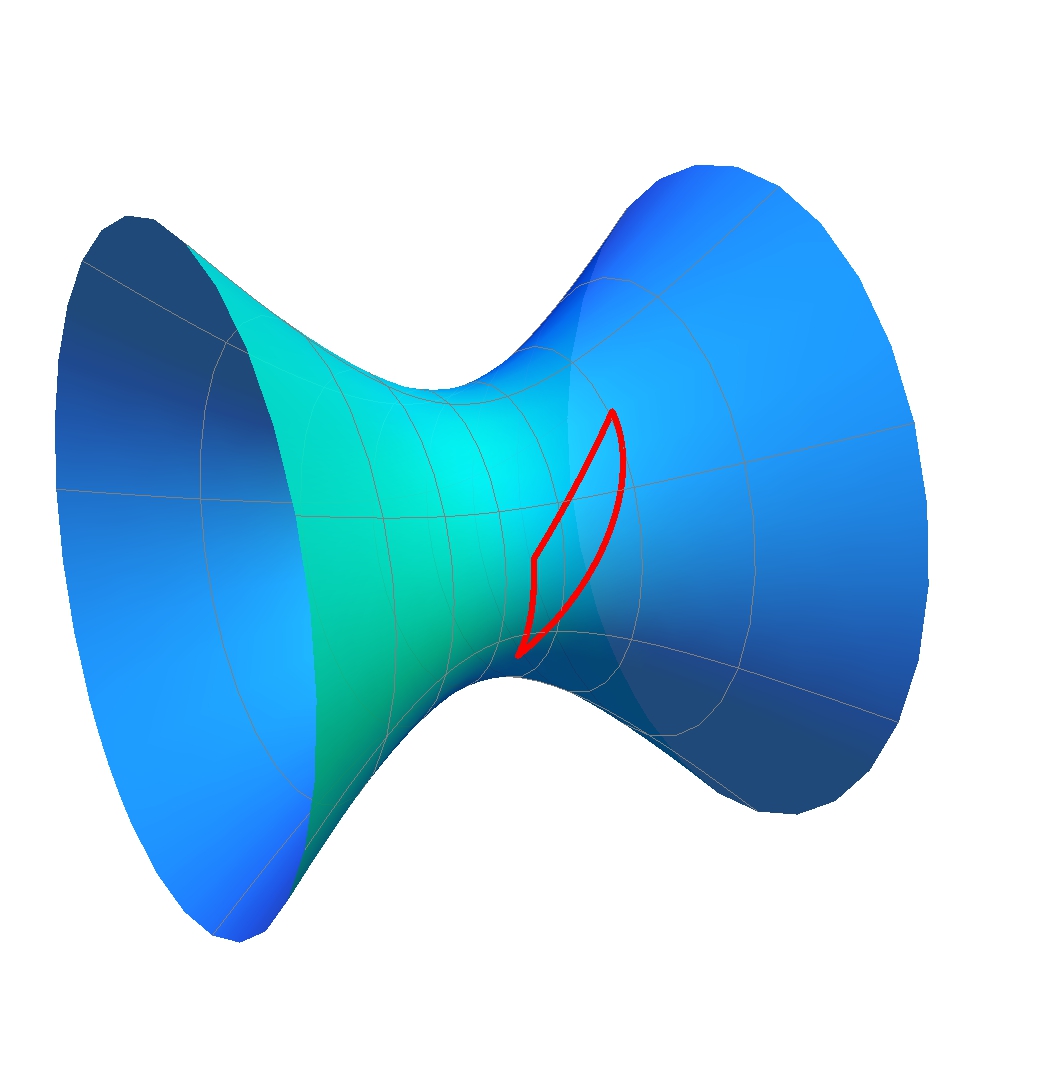}
\label{fg4302a}}
\quad
\subfigure[]{\includegraphics[width=4cm,height=3cm,keepaspectratio]{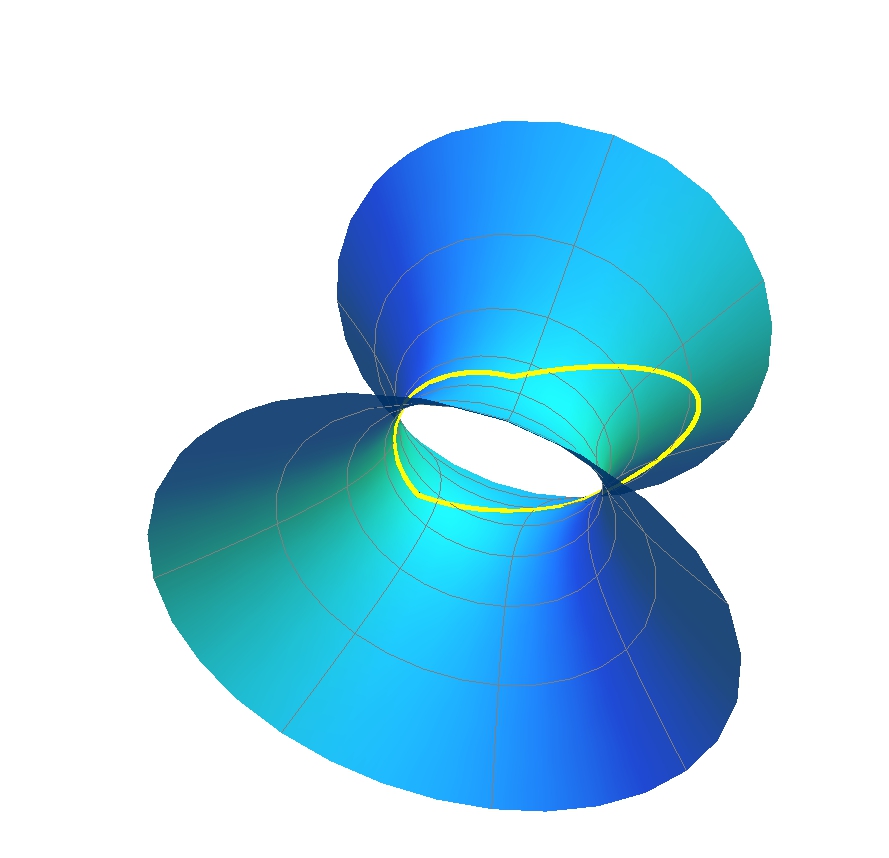}
\label{fg4302b}}
\quad
\subfigure[]{\includegraphics[width=4cm,height=3cm,keepaspectratio]{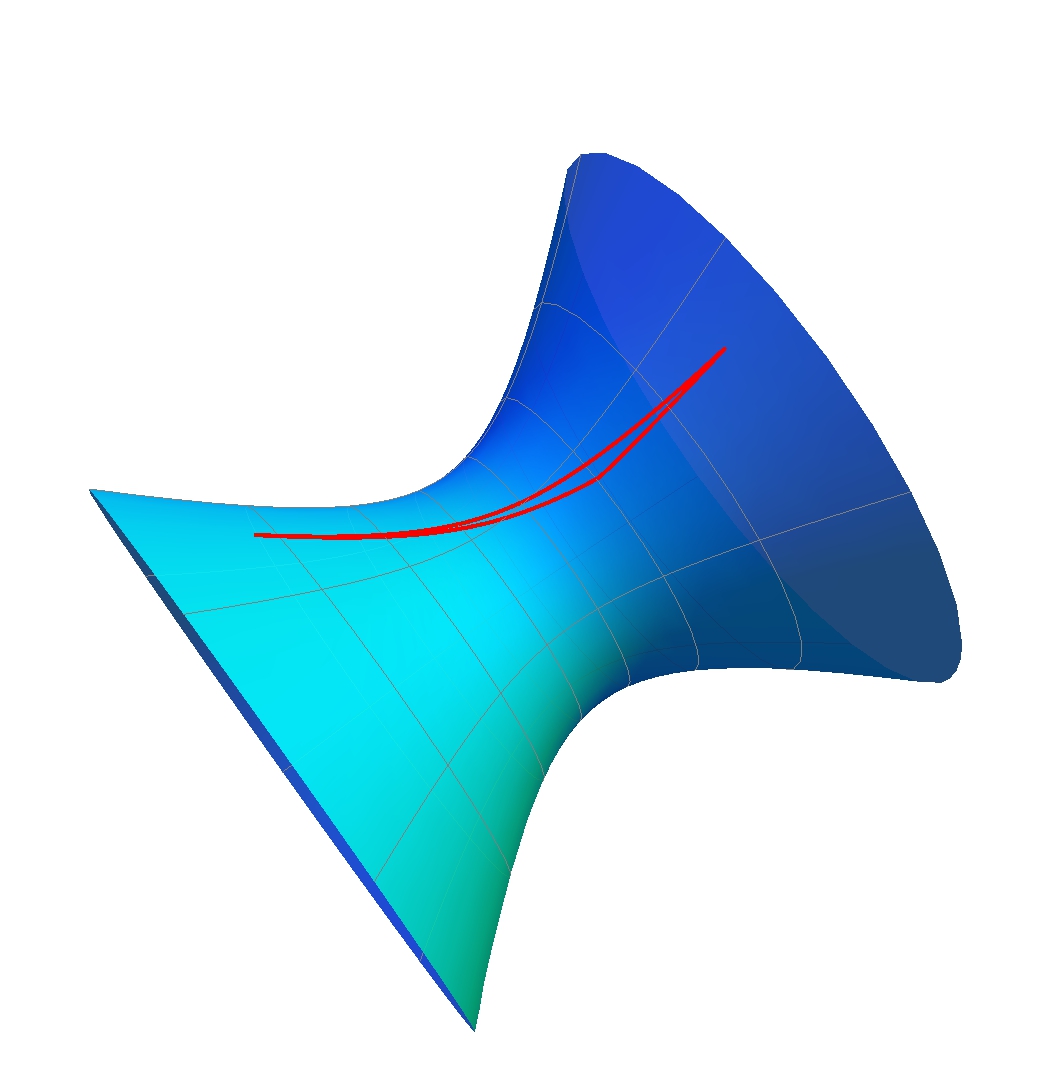}
\label{fg4302c}}
\quad
\subfigure[]{\includegraphics[width=4cm,height=3cm,keepaspectratio]{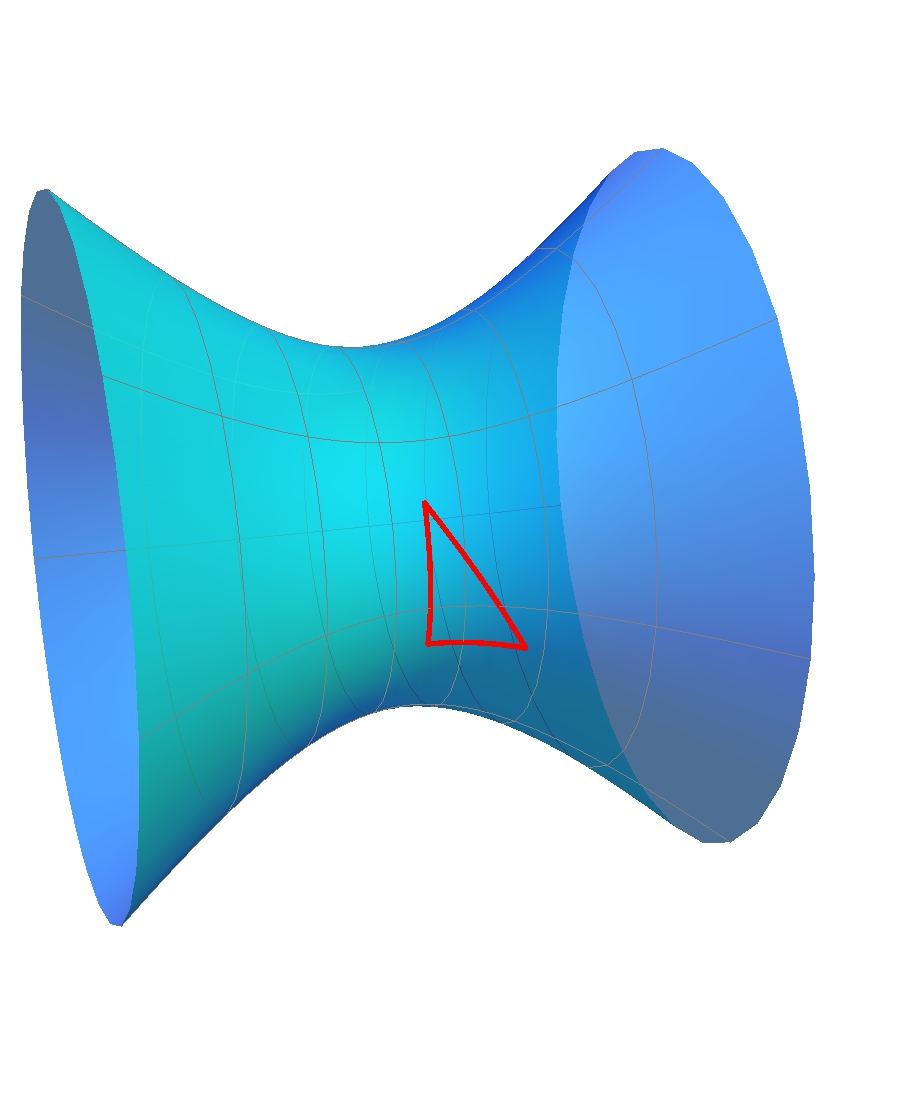}
\label{fg4302d}}
\quad
\subfigure[]{\includegraphics[width=4cm,height=3cm,keepaspectratio]{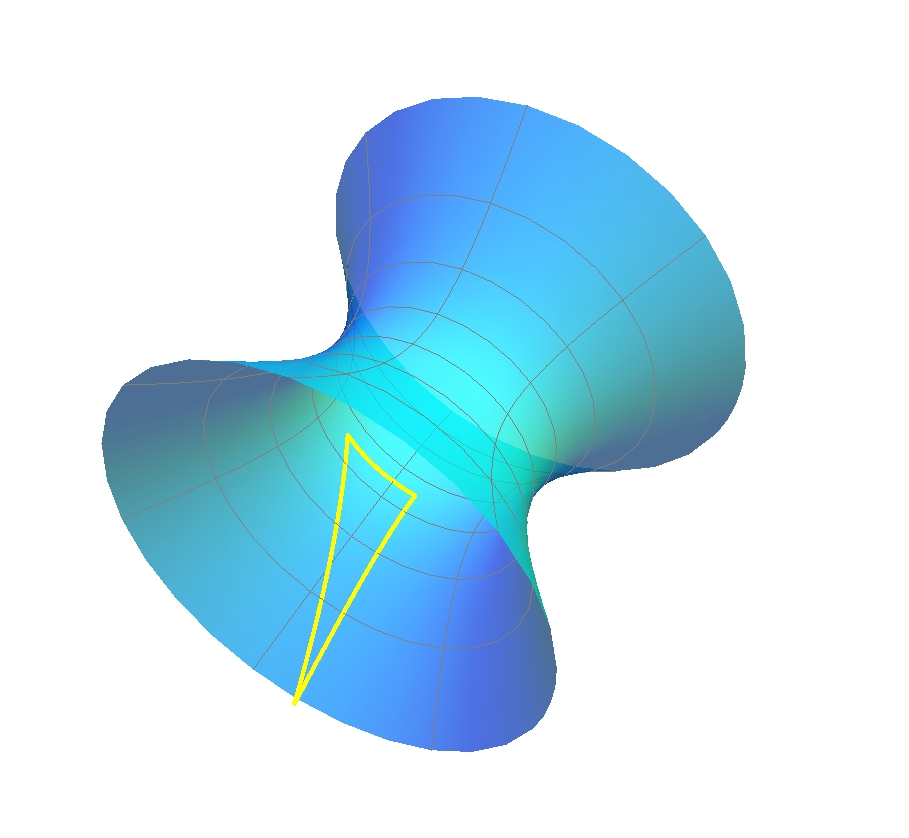}
\label{fg4302e}}
\caption{\scalefont{0.8}{(a) contractible spatiolateral (b) non-contractible spatiolateral (c) tempolateral (d) chorosceles (e) chronosceles}}
\label{fg4302}
\end{figure}
\newpage
\hspace*{\fill}

\subsection{ Pseudo-angle and Angle in Lorentz Space $\mathbb{R}_{1}^{3}$}
\hspace*{\fill}

By \cite{4}, we have the following definition
\begin{definition}\label{def1}
The pseudo-norm $\left\Vert u\right\Vert _{p}$ of $u\in \mathbb{R}_{1}^{3}$
is defined by the complex number
$$
\left\Vert u\right\Vert _{p}=\sqrt{\left\langle u,u\right\rangle }\in
R^{+}\cup \left\{ 0\right\} \cup R^{+}i,~i=\sqrt{-1}.
$$
\end{definition}
Then, we have
$$
\left\Vert u\right\Vert _{p}=\left\{
\begin{array}{ll}
0 & ,u\text{ null} \\
\sqrt{\left\vert \left\langle u,u\right\rangle \right\vert } & ,u\text{
space-like} \\
i\sqrt{\left\vert \left\langle u,u\right\rangle \right\vert } & ,u\text{
time-like}%
\end{array}%
\right.
$$
\begin{definition}\label{def2}
Let $u, v$ be unit non-null vectors in $\mathbb{R}_{1}^{3}$, then the complex
number $\phi (u,v)$ satisfying $\cos \phi =\dfrac{\left\langle
u,v\right\rangle }{\left\Vert u\right\Vert _{p}\left\Vert v\right\Vert _{p}}$
is called pseudo-angle between $u$ and $v$ (\cite{4} and \cite{9}).
\end{definition}
Let $u, v$ be unit non-null vectors, and let $U$  be the subspace
$sp\left\{ u,v\right\} $ of $\mathbb{R}_{1}^{3}$. Then we have the following definitions.

\begin{definition} \label{def4}
The angle $\theta$ between vectors $u$ and $v$ in $\mathbb{R}_{1}^{3}$  is given by
$$
\theta =\left\{
\begin{array}{ccclr}
\arccos\left(\left\langle u,v\right\rangle \right)&,&\text{U is space-like}&~&~\\
\arccos h\left( -\left\langle u,v\right\rangle \right)&,&\text{U is time-like and}&\left\langle u,u\right\rangle \left\langle v,v\right\rangle =1&\left\langle u,v\right\rangle <-1 \\
\arccos h\left( \left\langle u,v\right\rangle \right)&,&\text{U is time-like and}&\left\langle u,u\right\rangle \left\langle v,v\right\rangle =1&\left\langle u,v\right\rangle >1\\
\arcsin h\left( \left\langle u,v\right\rangle \right)&,&\text{U is time-like and}&\left\langle u,u\right\rangle \left\langle v,v\right\rangle =-1&~
\end{array}\right.
$$
\end{definition}
By Definition \ref{def2} and Definition \ref{def4}, we give $\phi$ depend on $\theta$
\begin{definition}\label{def5}

\hspace*{\fill}
\begin{enumerate}

\item
If $u,v$ are unit space-like vectors and $\theta >0$,then
$$
\phi =\left\{
\begin{array}{ll}
\pi-i\theta  & ,\left\langle u,v\right\rangle <-1 \\
i\theta  & ,\left\langle u,v\right\rangle >1 \\
\theta \in \left[ 0,\pi \right] & ,\left\langle u,v\right\rangle
\in \left[ -1,1\right]%
\end{array}%
\right.
$$
\item If $u,v$ are unit unit time-like vectors and $\theta >0$,then
$$
\phi =\left\{
\begin{array}{ll}
-i\theta  & ,\text{u and v are  same  time cone} \\
\pi+i\theta  & ,\text{u and v different time cone}
\end{array}%
\right.
$$
If $u$ unit space-like, $v$ unit time-like and $\theta \in\mathbb{R}$ then

$$
\phi=\frac{\pi }{2}+i\theta
$$
\end{enumerate}
\end{definition}

\section{Girard Type Theorems for Proper de Sitter Triangle with Non-null Edges}

Let $\triangle $ be a proper\ de Sitter\ triangle with  non null edge, and
let $p_{1}, p_{2}, p_{3}$ be vertices of $\triangle $. Let $
V_{j}^{k},V_{j}^{l} $ \ and $u_{j}$ be unit tangent vectors at vertex $p_{j}$
pointing in the direction vertices $p_{k}, p_{l}$ and the unit outer
normal to the edge plane opposite to vertex $p_{j}$, $j=1,2,3$. Then one can see that
\begin{equation}\label{eq4.1}
\langle V_{j}^{k},V_{j}^{l}\rangle =\langle u_{k},u_{l}\rangle, k\neq j\neq
l,~~~~j,k,l=1,2,3.
\end{equation}
By \cite[Remark 2.9]{3}, we have
$V_{j}^{k}$ is time-like (space-like) if and only if $u_{k}$ is space-like(time-like).
\begin{theorem}\label{teo6}
Let $\bigtriangleup $ be a triangle with vertices $p_{1}, p_{2}, p_{3}$ in
$S_{1}^{2}$, and let $\phi _{kl}$ be pseudo-angle between unit tangent vectors
$V_{j}^{k}$ and $V_{j}^{l}$
at vertex $p_{j}$ pointing in the direction vertices $p_{k}, p_{l}$ . Then
the area of $\bigtriangleup $ is
$$
\nabla=(\phi _{12}+\phi _{13}+\phi _{23}-\pi )\in\mathbb{R}^{+}i.
$$
\end{theorem}
\begin{proof}
See \cite[Theorem 5]{4}.
\end{proof}
\subsection{Girard's Theorem for $\protect\overset{0}{_{3}\triangle _{0}}$}

\hspace*{\fill}

$\overset{0}{_{3}\triangle _{0}}$\ is called contractible if the sum of lengths of edges is less then $2\pi$, non-contractible if greater then $2\pi$. The edges of non-contractible triangle are  satisfy triangle inequality while contractible one are not \cite{3}.

A non-contractible and contractible triangle has a polar triangle being {\it hyperbolic and strange triangle with one time like and  the other two edges are impossible}. Thus a contractible triangle has one and only one vertex at which the unit outer normals to the edge planes are in same time cone, but non-contractible spatiolateral triangle has three vertices at which the unit outer normals to the edge planes are in same time cone.

One can obtain a non-contractible spatiolateral triangle $\Omega^{\cdot }$ from a contractible spatiolateral de Sitter triangle $\Omega $ by taking the antipodal of vertex at which the unit outer normals to the edge planes are in same time cone.
\hspace*{\fill}

Let $\overset{0}{_{3}\triangle _{0}}$ be a spatiolateral  triangle with vertex set $\{p_{1},p_{2},p_{3}\}$, and let $V_{j}^{k}$ and $V_{j}^{l}$ be unit tangent vectors at vertex $ p_{j}$ pointing in the direction vertex $ p_{k}$ and $p_{l}$.Denote by $u_{j}$ the unit outer normal to the edge plane opposite to vertex $p_{j}$, $j=1,2,3$. Then $V_{j}^{k}, V_{j}^{l}$ are
space-like, and $u_{k}$, $u_{l}$ are time-like vectors. By equation (\ref{eq4.1}),
\begin{equation}\label{eq4.2}
cos\phi _{23}=\langle V_{1}^{2},V_{1}^{3}\rangle; cos\phi _{13}=\langle
V_{2}^{1},V_{2}^{3}\rangle; cos\phi _{12}=\langle V_{3}^{1},V_{3}^{2}\rangle
\end{equation}
where $\phi _{kl}$ is the pseudo-angle at vertex $ p_{j}$  of $\ \overset{0}{_{3}\triangle _{0}}$.
\begin{theorem}\label{teo7}
Let $\overset{0}{_{3}\triangle _{0}}$ be contractible spatiolateral de
Sitter triangle with the measure of interior angles
$\theta _{23}, \theta _{12}, \theta _{13}$ and let $\theta
_{23}$ be interior angle at vertex $p_{1}$ which the unit outer normals to the edge planes are in same time cone.
Then the area V of  $\overset{0}{_{3}\triangle _{0}}$ is
\begin{equation*}
V=-\theta _{23}+\theta _{12}+\theta _{13}.
\end{equation*}
\end{theorem}
\begin{proof}
By \cite[Corollary 4.10]{3}, contractible spatiolateral triangle has a
strange polar triangle whose two vertices in the same time cone other one in
different time cone. Then by equation (\ref{eq4.1}),
$$
\langle V_{1}^{2},V_{1}^{3}\rangle
 <-1~~,~~\langle V_{2}^{1},V_{2}^{3}\rangle >1~~,~~\langle V_{3}^{1},V_{3}^{2}\rangle >1.
$$
By Definition \ref{def5}, we have
\begin{equation*}\label{eq4.3}
\phi _{23}=\pi-i{\theta _{23}}~~,~~\phi _{13}=i{\theta _{13}}~~,~~ \phi
_{12}=i{\theta _{12}}.
\end{equation*}
By Theorem \ref{teo6}, we obtain
\begin{equation}
\nabla =i(-{\theta _{23}}+\theta _{13}+\theta _{12}),
\end{equation}
 which is completes the proof.
\end{proof}
By equation (\ref{eq4.3}) and Definition \ref{def4}, we see that
\begin{equation*}
\theta _{23}=\arccos h(-\langle V_{1}^{2},V_{1}^{3}\rangle )~~,~~\theta _{13}%
 =\arccos h(\langle V_{2}^{1},V_{2}^{3}\rangle )~~,~~\theta _{12} =\arccos
h(\langle V_{3}^{1},V_{3}^{2}\rangle ).
\end{equation*}
By Theorem \ref{teo7}, the following corollary has been proved.
\begin{corollary}
If $u_{j}$ is the unit outer normal to the edge plane
opposite to vertex $p_{j}$ of contractible spatiolateral triangle $\overset{0}{_{3}\triangle _{0}}$, then
\begin{equation*}
V=-\arccos h(-\langle V_{1}^{2},V_{1}^{3}\rangle ))+\arccos h(\langle V_{2}^{1},V_{2}^{3}\rangle )+\arccos
h(\langle V_{3}^{1},V_{3}^{2}\rangle ).
\end{equation*}
\end{corollary}
\begin{remark}
By definition of non-contractible spatiolateral triangle, one can easily see
that it is not restrict an area on de Sitter plane.
\end{remark}
\subsection{Girard's Theorem for  $\protect\overset{0}{_{0}\triangle _{3}}$}
\hspace*{\fill}

Let $\overset{0}{_{0}\triangle _{3}}$ be a tempolateral de Sitter
triangle with vertices $p_{1}, p_{2}, p_{3}.$ Then by \cite[Lemma 3.3]{3}, $\overset{0}{_{0}\triangle _{3}}$ has one and only one vertex at which the time-like unit tangent vectors are in different time cone. No loss of
generality we choose that vertex $p_{1}$. By \cite[Theorem 5.9 ]{3}, the
interior angle ${\theta _{23}}$ at vertex $p_{1}$ is greater then the
sum of other two interior angle $\theta _{12}, \theta _{13}$ of  $\overset{0}{_{0}\triangle _{3}}$. That is
$$
\theta _{23}-\theta _{12}-\theta _{13}>0.
$$
Let \ V$_{j}^{k}$ be unit tangent vector at vertex $p_{j}$ pointing in the
direction vertex $p_{k}$. Then by equation (\ref{eq4.1}), we have
\begin{equation}\label{eq4.4}
\left\langle V_{1}^{2},V_{1}^{3}\right\rangle >1~~,~~\left\langle V_{2}^{1},V_{2}^{3}\right\rangle <-1~~,~~\left\langle V_{3}^{1},V_{3}^{2}\right\rangle <-1 .
\end{equation}
By Definition \ref{def5}, we have
\begin{equation}\label{eq4.5}
\phi _{23}=\pi +i{\theta _{23}}~~,~~\phi _{13}=-i{\theta _{13}}~~,~~\phi _{12}=-i{\theta _{12}}.
\end{equation}
By Theorem \ref{teo6}, we have
\begin{equation*}
\nabla =i(\theta _{23}-\theta _{12}-\theta _{13})\in {R}^{+}{i}.
\end{equation*}
Therefore, we have been proved the following theorem.
\begin{theorem}\label{teo8}
Let $\overset{0}{_{0}\triangle _{3}}(p_{1},p_{2},p_{3})$ be  a
tempolateral de Sitter triangle with the greater angle at vertex $p_{1}$.
Then the area $V$ of $\overset{0}{_{0}\triangle _{3}}(p_{1},p_{2},p_{3})$ is
\begin{equation*}
V=\theta _{23}-\theta _{12}-\theta _{13}.
\end{equation*}
\end{theorem}
By Definition \ref{def4}  and equation (\ref{eq4.4}), we have
\begin{equation*}
\theta _{23}= \arccos h (\left\langle V_{1}^{2},V_{1}^{3}\right\rangle), \theta _{13}=\arccos h (-\left\langle V_{2}^{1},V_{2}^{3}\right\rangle), \theta _{12}=\arccos h (-\left\langle V_{3}^{1},V_{3}^{2}\right\rangle) .
\end{equation*}
Now we have the following result.
\begin{corollary}
Let $\overset{0}{_{0}\triangle _{3}}(p_{1},p_{2},p_{3})$ be  a tempolateral
de Sitter triangle with the greater angle at vertex $p_{1}$.
Then the area $V $ of $\overset{0}{_{0}\triangle _{3}}(p_{1}, p_{2}, p_{3})$ is given by
$$
V=\arccos h (\left\langle V_{1}^{2},V_{1}^{3}\right\rangle)-\arccos h (-\left\langle V_{2}^{1},V_{2}^{3}\right\rangle)-\arccos h (-\left\langle V_{3}^{1},V_{3}^{2}\right\rangle).
$$
\end{corollary}
\subsection{Girard's Theorem for $\protect\overset{0}{_{2}\triangle _{1}}$}

\hspace*{\fill}

Let $\overset{0}{_{2}\triangle _{1}}$ be chorosceles de Sitter triangle with
vertices $p_{1}, p_{2}, p_{3}$ and let $V_{i}^{k}$ and $V_{i}^{l}$ be tangent
vectors at vertex $p_{i}$ pointing in the direction vertex $p_{k}$ and $
p_{l}$. Denote by $u_{i}$ the unit outer normal to the edge plane opposite to vertex $p_{i}$, $i=1,2,3$. Without no loss of generality, we
choose the time-like edge opposite to vertex $p_{1}$. Then by Theorem \ref{teo5}, we have
\begin{equation}\label{eq4.7}
\left\vert \left\langle p_{1},p_{2}\right\rangle \right\vert <1~~,~~\left\vert
\left\langle p_{1},p_{3}\right\rangle \right\vert <1~~,~~\left\langle
p_{2},p_{3}\right\rangle >1
\end{equation}
Since the polar triangle of $\overset{0}{_{2}\triangle _{1}}$ is strange triangle, we have
\begin{equation*}
\left\langle V_{1}^{2},V_{1}^{3}\right\rangle =\langle u_{2},u_{3}\rangle >1
\end{equation*}
or in another way, $u_{2}$ and $u_{3}$ are in different time cone.
Thus the polar triangle of $\overset{0}{_{2}\triangle _{1}}$ has  vertices as follows
$u_{1}\in S_{1}^{2}$ and $u_{2}, u_{3}\in H^{2}\cup (-H^{2})$.
\begin{theorem}\label{teo9}
Let $\overset{0}{_{2}\triangle _{1}}$ be chorosceles de Sitter triangle
with interior angles $\theta _{23}, \theta _{13}$ and $\theta _{12}$. Then
the area $V$  of $\overset{0}{_{2}\triangle _{1}}$ is given by
\begin{equation*}
V=\theta _{23}+\theta_{12}+\theta _{13}.
\end{equation*}
\end{theorem}
\begin{proof}
By $u_{2}, u_{3}$ are time-like vectors in different time cone and Definition \ref{def5}, we see that
$$
\phi _{23}=i\theta _{23},~ \phi _{13}=\frac{\pi }{2}+i\theta _{13},~ \phi _{12}=\frac{\pi }{2}+i\theta _{12}
$$
By Theorem \ref{teo6}, we obtain
\begin{equation*}
\nabla =i(\theta _{23}+\theta _{13}+\theta _{12}).
\end{equation*}
This completes the proof.
\end{proof}
By Definition \ref{def4}, we have
\begin{equation*}
\theta _{23}={arccosh}(\left\langle V_{1}^{2},V_{1}^{3}\right\rangle ),~
\theta _{13}={arcsinh}(\left\langle V_{2}^{1},V_{2}^{3}\right\rangle ),~\theta _{12}={arcsinh}(\left\langle V_{3}^{1},V_{3}^{2}\right\rangle).
\end{equation*}
By using these equations in Theorem \ref{teo9}, we obtain the following corollary.
\begin{corollary}
The area $V$ of chorosceles de Sitter triangle $\overset{0}{_{2}\triangle
_{1}}$ is given by
\begin{equation*}
V={arccosh}(\left\langle V_{1}^{2},V_{1}^{3}\right\rangle+{arcsinh}(\left\langle V_{2}^{1},V_{2}^{3}\right\rangle )+{arcsinh}(\left\langle V_{3}^{1},V_{3}^{2}\right\rangle)).
\end{equation*}
\end{corollary}
\subsection{Girard's Theorem for $\protect\overset{0}{_{1}\triangle _{2}}$}
\hspace*{\fill}

Let $\overset{0}{_{1}\triangle _{2}}$ be chronosceles de Sitter triangle
with vertices $p_{1}, p_{2}, p_{3}$ and let $V_{j}^{k}$ and $V_{j}^{l}$ be
unit tangent vectors at vertex $p_{j}$ pointing in the direction vertex $p_{k}$
and $p_{l}$. Denote by $u_{j}$ the unit outer normal to the edge plane opposite to vertex $p_{j}$, $j=1,2,3$. Without no loss of generality, we
choose the space-like edge opposite to vertex $p_{1}$. Then by Theorem \ref{teo4}, we have
\begin{equation*}
\left\vert \left\langle p_{2},p_{3}\right\rangle \right\vert <1~~,~~\left\langle
p_{1},p_{2}\right\rangle >1~~,~~\left\langle p_{1},p_{3}\right\rangle >1.
\end{equation*}
The polar triangle of $\overset{0}{_{1}\triangle _{2}}$ is strange triangle
with space-like edge bounded by $u_{2}, u_{3}\in S_{1}^{2}$. Therefore we have
$$
\left\langle V_{1}^{2},V_{1}^{3}\right\rangle =\langle u_{2},u_{3}\rangle
<-1.
$$
By Definition \ref{def5}, we have
$$
\phi _{23}=-i\theta _{23}~~,~~\phi _{12}=\frac{\pi }{2}+i\theta_{12}~~,~~\phi _{13}=\frac{\pi }{2}+i\theta _{13} .
$$
By Theorem \ref{teo6}, we obtain
\begin{equation*}
\nabla =i(-\theta _{23}+\theta _{12}+\theta _{13}).
\end{equation*}
So, we have been proved the following theorem.
\begin{theorem}\label{teo10}
Let $\overset{0}{_{1}\triangle _{2}}$ be choronosceles de Sitter triangle
with interior angles $\theta _{23}, \theta _{13}$ and $\theta _{12}$. Then
the area $V$ of $\overset{0}{_{1}\triangle _{2}}$ is given by
\begin{equation*}
V=-\theta _{23}+\theta _{12}+\theta _{13}.
\end{equation*}
\end{theorem}

By Definition \ref{def4}, we have
$$
\theta _{23}={arccosh}(-\left\langle V_{1}^{2},V_{1}^{3}\right\rangle),~\theta _{13}={arcsinh}(\left\langle V_{2}^{1},V_{2}^{3}\right\rangle),~\theta _{12}={arcsinh}(\left\langle V_{3}^{1},V_{3}^{2}\right\rangle).
$$
Then, we obtain the following corollary.
\begin{corollary}
The area $V$ of choronosceles de Sitter triangle $\overset{0}{_{1}\triangle
_{2}}$ is given by
\begin{equation*}
V =-{arccosh}(-\left\langle V_{1}^{2},V_{1}^{3}\right\rangle)+{arcsinh}(\left\langle V_{2}^{1},V_{2}^{3}\right\rangle)+{arcsinh}(\left\langle V_{3}^{1},V_{3}^{2}\right\rangle).
\end{equation*}
\end{corollary}

\end{document}